\documentclass[journal]{IEEEtran}
\usepackage{amsmath,amsfonts,amsthm}
\usepackage{amssymb}
\usepackage{algorithmic}
\usepackage{algorithm}
\usepackage{array}
\usepackage{textcomp}
\usepackage{stfloats}
\usepackage{url}
\usepackage{verbatim}
\usepackage{graphicx}
\usepackage{subfigure}
\usepackage{cite}
\usepackage{color}
\usepackage{bm}
\usepackage{mathrsfs}
\newtheorem{theorem}{Theorem}
\newtheorem{lemma}{Lemma}
\begin{document}
\title{Asymptotic Performance Analysis of Fluid Antenna Systems: An Extreme Value Theory Perspective
}

\author{Yi Zhang, Jintao Wang, Zheng Shi, Xu Wang, Guanghua Yang, Shaodan Ma, and Kai-Kit Wong
\thanks{Yi Zhang, Jintao Wang, Zheng Shi, Xu Wang, and Guanghua Yang are with the School of Intelligent Systems Science and Engineering, Jinan University, Zhuhai, 519070, China (e-mails: yizhang@stu2023.jnu.edu.cn; wang.jintao@connect.um.edu.mo; zhengshi@jnu.edu.cn; xuwang@jnu.edu.cn; ghyang@jnu.edu.cn).}
\thanks{Shaodan Ma is with Faculty of Science and Technology, University of Macau, Macau, 999078, China (e-mail: shaodanma@um.edu.mo).}
\thanks{Kai-Kit Wong is affiliated with the Department of Electronic and Electrical Engineering, University College London, United Kingdom (e-mail: kai-kit.wong@ucl.ac.uk).}}

%
\maketitle
\begin{abstract}
Fluid antenna systems (FAS) allow dynamic reconfiguration to achieve superior diversity gains and reliability. To quantify the performance scaling of FAS with a large number of antenna ports, this paper leverages extreme value theory (EVT) to conduct an asymptotic analysis of the outage probability (OP) and ergodic capacity (EC). The analysis reveals that the OP decays approximately exponentially with the number of antenna ports. Moreover, we establish upper and lower bounds for the asymptotic EC, uncovering its double-logarithmic scaling law. Furthermore, we re-substantiate these scaling laws by exploiting the fact that the mode of the Gumbel distribution scales logarithmically. Besides, we theoretically prove that spatial correlation among antenna ports degrades both OP and EC. All analytical findings are conclusively validated by numerical results.
\end{abstract}


\section{Introduction}

\IEEEPARstart{F}{luid} antenna technology has emerged as a promising paradigm for future wireless systems, offering the ability to dynamically reconfigure the physical location and geometry of the radiating element. This flexibility provides significant advantages in spatial diversity over conventional fixed-position antennas \cite{10303274}. Owing to its potential for improving spectral efficiency, reliability, and latency, fluid antenna systems (FAS) are considered a key enabler for sixth-generation (6G) networks \cite{10753482}.

In recent years, considerable research efforts have been devoted to evaluating the reliability and spectral efficiency of FAS, with particular emphasis on the outage probability (OP) and ergodic capacity (EC) \cite{10130117,11023237,10564133,9264694,10924151,11039166,11184593}. The OP, as a fundamental measure of link reliability, has been extensively studied. For instance, the work in \cite{10130117} derived an approximate distribution for the equivalent channel gain in FAS, facilitating OP evaluation. Further, an asymptotic OP expression in the high-SNR regime was obtained in \cite{11023237} by approximating the spatial correlation matrix among fluid antenna ports. Departing from the widely adopted Jake’s model, the authors of \cite{10564133} employed copula theory to model the joint distribution of fading channels across ports, capturing both linear and non-linear correlation structures. This approach led to exact and asymptotic OP expressions under generalized correlation models. As a complementary performance metric, the EC of FAS has also attracted attention \cite{10924151,11039166,11184593}. Specifically, analytical frameworks originally developed for OP have been extended to characterize the EC in fluid antenna multiple access (FAMA) systems \cite{11039166} and reconfigurable intelligent surface (RIS)-aided multiuser FAS \cite{11184593}, particularly in the high-SNR regime. However, existing literature often yields complex expressions that obscure clear insights, especially when the number of antenna ports is large and spatial correlation is considered. Consequently, the scaling behavior of OP and EC with respect to the number of ports remains inadequately understood.

To address this gap, this paper leverages extreme value theory (EVT) to derive the asymptotic OP and EC under a large number of antenna ports. Our analysis reveals that the OP decays approximately exponentially with the number of ports. Furthermore, we derive upper and lower bounds for the asymptotic EC, demonstrating its double-logarithmic scaling with the number of antenna ports. These scaling laws are also verified by analyzing the scaling behavior of the Gumbel distribution's mode. Additionally, we rigorously prove that spatial correlation among antenna ports adversely affects both OP and EC performance.

The remainder of the paper is structured as follows. Section \ref{sec_model} introduces the FAS system model. Section \ref{sec_per} presents the asymptotic analysis of the OP and EC. Numerical results are provided in Section \ref{sec_num}, and conclusions are drawn in Section \ref{sec_con}.

\section{System Model}\label{sec_model}
We consider a point-to-point system where a single-antenna Base Station (BS) communicates with a user equipped with a fluid antenna. The FAS consists of $N$ uniformly spaced antenna ports along a linear region of length $W\lambda$, where $\lambda$ is the carrier wavelength and $W$ is the normalized length. The received signal at the $n$-th port is given by
\begin{align}\label{sys}
\mathbf{y}_n=\sqrt{P}h_n\mathbf{x}+\mathbf{z}_n,
\end{align}
where $P$ is the transmit power, $\mathbf{x}$ is the transmitted sequence, $h_n$ is the fading channel coefficient from the BS to the $n$-th port, and $\mathbf{z}_n \sim \mathcal{CN}(0,\delta^2 {\bf I})$ is the additive white Gaussian noise (AWGN). Hence, the instantaneous signal-to-noise ratio (SNR) at port $n$ is $\gamma_n = |h_n|^2 \bar{\gamma}$, where $\bar{\gamma} = P/\delta^2$ denotes the average transmit SNR.

The FAS operates by selecting the port that maximizes the instantaneous SNR, following a selection combining scheme. The equivalent received SNR $\gamma=\bar{\gamma} g$ is determined by the maximum instantaneous channel power gain $g=\max\nolimits_{n\in\{1,\cdots,N\}} {|h_n|^2}$. 
\subsection{Correlation Channel Model}

Due to the extremely close proximity of the ports, the channel coefficients $\{h_n\}$ are highly correlated. To enable a mathematically tractable asymptotic analysis, we adopt Wong's correlated channel model, where the channel coefficient $h_n$ is expressed as \cite{wong2022closed}
\begin{align}\label{ch-co}
h_n=\sigma\left( \sqrt{1-\mu^2}\vartheta_{n}+\mu\vartheta_0\right),
\end{align}
where $\vartheta_{n}$ and $\vartheta_0$ are independent circularly-symmetric complex Gaussian random variables with zero mean and unit variance, i.e., $\vartheta_{n},\vartheta_0\sim\mathcal{CN}(0,1)$ and $\Bbb{E}(|h_n|^2)=\sigma^2$. The correlation among fading channels can be characterized by the parameter $\mu$, which is given by
\begin{align}\label{mu_W}
\mu = \sqrt{2{}_{1}F_{2}\left( {1}/{2};1;{3}/{2};-\pi^{2}W^{2}\right) -{J_{1}\left( 2\pi W\right) }/({\pi W})},
\end{align}
where $_{1}F_{2}\left(\cdot;\cdot;\cdot;\cdot\right)$ represents the generalized hypergeometric function and $J_1(\cdot) $ denotes the first-order Bessel function of the first kind.
\subsection{Performance Metrics}
We consider two performance metrics in this letter, including outage probability and ergodic capacity. Their definitions are explicitly given as follows.
\subsubsection{Outage Probability (OP)} The OP is defined as the probability that the mutual information $\mathcal{I}$ falls below a target rate $R$, i.e.,
\begin{equation}\label{eqn:p_out_def}
    P_{\rm out}=\Pr\left( \mathcal{I}<R\right)=F_g\left({(2^R-1)}/{\bar{\gamma}}\right),
\end{equation}
where $\mathcal{I}=\log_2(1+\bar{\gamma}g)$ denotes the mutual information and $F_g(\cdot)$ is the cumulative distribution function (CDF) of $g$.
\subsubsection{Ergodic Capacity (EC)}
The EC of the considered FAS can be calculated as
\begin{equation}\label{er_def}
    \bar{C}=\Bbb{E}(\mathcal{I})=\int_{0}^{\infty} \log_2(1+\bar{\gamma}x)f_{g}(x){\rm d} x,
\end{equation}
where $f_g(\cdot)$ denotes the probability density function (PDF) of $g$. 

Clearly, in order to derive the OP and the EC, it boils down to deriving the distribution of the equivalent channel gain $g$. However, determining the exact distribution $F_g(x)$ or $f_g(x)$ is generally intractable due to the intricate spatial correlation among the $N$ ports, especially as $N \to \infty$. Therefore, we focus on characterizing the asymptotic distribution of $g$ using EVT to derive novel scaling laws and fundamental performance limits.

\section{EVT Based Performance Analysis}\label{sec_per}
This section employs EVT to obtain the asymptotic expressions for OP and EC under the condition of a large number of ports, i.e., $N\to \infty$, with which scaling laws of the OP and the EC can be revealed.
\subsection{OP}
The correlation among channel gains $\{|h_n|^2,n\in [1,N]\}$ hinders the derivation of the OP in \eqref{eqn:p_out_def}. Fortunately, we note that the channel gains $\{|h_n|^2,n\in [1,N]\}$ are conditionally independent non-central chi-squared distributed with $2$ degrees of freedom, conditioned on the common channel parameter $g_0\triangleq{|\vartheta_0|^2}$ \cite[Th.1.3.4]{simon2004digital}. The conditional PDF and CDF of $\{|h_n|^2,n\in [1,N]\}$ are given below.
\begin{lemma}\label{le1_ind}
    The conditional PDF and CDF of $|h_n|^2$ given $g_0=|\vartheta_0|^2$ is written as 
    \begin{align}\label{h_pdf}
    f_{|h_n|^2|g_0}\left(\left.x_n\right|t\right)&=\frac{1}{{\sigma^2(1-\mu^2)}}e^{{-\frac{{\sigma^2}{\mu}^2t+x_n}{\sigma^2(1-\mu^2)}}}I_0\bigg {(} \frac{2\sqrt{{\mu}^2x_nt}}{\sigma(1-\mu^2)}\bigg{)},
    \end{align}
    \begin{align}\label{h_cdf}
    F_{|h_n|^2|g_0}\left(\left.x_n\right|t\right)&=1-Q_1\left(\sqrt{\frac{2\mu^2t}{1-\mu^2}},\sqrt{\frac{2x_n}{\sigma^2(1-\mu^2)}}\right),
    \end{align}
    where $g_0$ obeys the exponential distribution, i.e., $f_{g_0}(t)=e^{-t}$ for $t\ge 0$, $I_0(\cdot)$ denotes the modified Bessel function \cite[Eq.9.6.10]{abramowitz1968handbook}, and $Q_1(\cdot)$ is the first-order Marcum Q-function.
\end{lemma}

By invoking the independence among  $|h_n|^2$ given $g_0$ stated in Lemma \ref{le1_ind}, the OP in \eqref{eqn:p_out_def} can be further rewritten as 
\begin{align}\label{eqn:p_out_der}
    P_{\rm out}&={\mathbb E_{{g_0}}}\left( {\Pr \left( {\left. { {\mathop {\max }\limits_{n \in \{ 1, \cdots ,N\} } |{h_n}{|^2}} < \frac{{{2^R} - 1}}{{\bar \gamma }}} \right|{t}} \right)} \right),
\end{align}
where $\mathbb E(\cdot)$ denotes the expectation operator. Thus, \eqref{eqn:p_out_der} turns to determining the extreme value of channel gains $|h_n|^2$ given $g_0$. To characterize the asymptotic behavior of the equivalent channel gain $g$ under a large number of ports at FAS, the conditional EVT is leveraged, as introduced in the following lemma. 
\begin{lemma}\label{le2_ind}
    Let $z_1, z_2, \cdots, z_N$ be independent and identically distributed (i.i.d.) random variables with the CDF $F(z)$ and the PDF $f(x)$. The asymptotic distribution of the maximal value $z_{\max} = \max\{z_1, z_2, \cdots, z_N\}$ exists and converges to the Gumbel distribution $F_{z_{\max}}(z) \simeq \exp(-\exp(-(z-l_N)/a_N))$, if $f(z) > 0$, $f(\cdot)$ is differentiable on the interval $(z_0, \infty)$ for some $z_0$, and the following conditions holds\cite{david2004order}
    \begin{align}\label{le2_lim}
        \lim_{z\to\infty}{\frac{\mathrm{d}}{\mathrm{d}z}}\frac{1-F\left(z\right)}{f(z)}=0.
    \end{align}
    Furthermore, the location and the scale parameters are chosen as $l_N=F^{-1}\left(1-{1}/{N}\right)$ and $a_N=(Nf(l_N))^{-1}>0$.  
\end{lemma}
According to {{Lemma \ref{le2_ind}}}, the conditional distribution of the equivalent channel gain $g$ is given in the following theorem.
\begin{theorem}\label{th_h_dis}
   The asymptotic conditional distribution of the channel gain $g$ given $g_0 $ as $N\to \infty$ follows a Gumbel distribution with the conditional CDF and PDF as
\begin{align}\label{eqn:f_cdf}
F_{g|g_0}(x|t)&\simeq\exp{\left(-e^{-\frac{x-l_{N,t}}{a_{N,t}}}\right)},
\end{align}
\begin{align}\label{eqn:f_pdf}
f_{g|g_0}(x|t)&\simeq\frac{1}{a_{N,t}}\exp{\left(-e^{-\frac{x-l_{N,t}}{a_{N,t}}}-\frac{x-l_{N,t}}{a_{N,t}}\right)}.
\end{align}
   where the location and the scale parameters $l_{N,t}$ and $a_{N,t}$ satisfy $F_{|h_n|^2|g_0}\left(\left.l_{N,t}\right|t\right)=1-{1}/{N}$ and $a_{N,t}=(Nf_{|h_n|^2|g_0}(l_{N,t}|t))^{-1}$. Besides, the first three conditional moments of $g$ given $g_0$ are got by $\Bbb{E}(g|g_0=t)=l_{N,t}+\Upsilon  a_{N,t}$, $\Bbb{E}(g^2|g_0=t)=(l_{N,t}+a_{N,t}\Upsilon )^2+\pi^2a_{N,t}^2/6$, and $\Bbb{E}(g^3|g_0=t)=(l_{N,t}+a_{N,t}\Upsilon )^3+\pi^2a_{N,t}^2(l_{N,t}+a_{N,t}\Upsilon )/2+2a_{N,t}^3\zeta(3)$, where $\zeta(\cdot)$ and $\Upsilon =0.5772$  denote Riemann zeta function and the Euler–Mascheroni constant, respectively. Consequently, the variance of $g$ conditioned on $g_0$ is obtained as ${\rm{Var}}(g|g_0=t)={(a_{N,t}\pi)^2}/{6}$. 
    Furthermore, the asymptotic expressions of $l_{N,t}$ and $a_{N,t}$ can be represented as $l_{N,t}\simeq\sigma^2(1-\mu^2)\ln N$ and $a_{N,t}\simeq \sqrt{4\pi\mu\sigma^4(1-\mu^2)^{3/2}}(t\ln N)^{1/4}$ as $N\to \infty$. In addition, the mode of the conditional channel gain $g$ given $g_0$ is $l_{N,t}$ and $\mathbb E(g|g_0) \ge {\rm Mo}(g|g_0)$, where ${\rm Mo}(g|g_0) =  \arg\max_{x} f_{g|g_0}(x|t)$.
\end{theorem}
\begin{proof}
    Please see Appendix \ref{sec_th_1}.
\end{proof}

By substituting \eqref{eqn:f_cdf} into \eqref{eqn:p_out_der}, the outage probability is asymptotic to
\begin{align}\label{op}
P_{\rm out}
\simeq{\mathbb E_{{g_0}}}\left(F_{g|g_0}(x|t)\right)=\int_{0}^\infty \exp{(-e^{-\frac{\frac{2^R-1}{\bar{\gamma}}-l_{N,t}}{a_{N,t}}}-t)}   {\rm d}t.
\end{align}

By using the asymptotic expressions of $l_{N,t}$ and $a_{N,t}$ in Theorem \ref{th_h_dis}, we arrive at
\begin{equation}\label{eqn:out_saclinglaw}
    P_{\rm out} \simeq {\mathbb E_{{g_0}}}\left(\exp \left( { - {e^{{\frac{{{{\left( {\ln N} \right)}^{3/4}}}}{{\sqrt {4\pi \mu {{(1 - {\mu ^2})}^{ - 1/2}}} {t^{1/4}}}}}}}} \right) \right).
\end{equation}
It is found that the OP decreases with $(\ln N)^{3/4}$ in double-exponential scaling law, i.e., $P_{\rm out} = \exp(-\exp({\rm const} \cdot (\ln N)^{3/4}))$. This conclusion can be roughly drawn using the result of mode, namely ${\rm Mo}(g|g_0) = l_{N,t} \propto \ln N$. Furthermore, OP is an increasing function of $\mu$, which justifies the adverse impact of spatial correlation between antenna ports.

\subsection{EC}\label{sec_ec}
By using the conditional independence of $|h_1|^2,\cdots,|h_N|^2$ given $g_0$, as in \eqref{eqn:p_out_der}, we get a new representation of \eqref{er_def} as 
\begin{align}\label{eqn:er_der}
    {\bar{C}}&=\Bbb{E}_{g_0}\left( \Bbb{E}_{g|g_0}(\log_2(1+\bar{\gamma}x)|t)\right).
\end{align}
By using Theorem \ref{th_h_dis} and the exponential distribution of $g_0$, the asymptotic EC $\bar{C}_{\rm asy}$ is derived as
    \begin{align}\label{eqn:er}
    &\bar{C}_{\rm asy}=\frac{1}{a_{N,t}}\int\limits_{0}^\infty\int\limits_{0}^{\infty}  {\log_2(1+\bar{\gamma}x)}e^{-e^{-\frac{x-l_{N,t}}{a_{N,t}}}-\frac{x-l_{N,t}}{a_{N,t}}-t}    {\rm d} x {\rm d}t.
    \end{align}
However, due to the implicit expression of $l_{N,t}$ and $a_{N,t}$, it is challenging to gain more insights from \eqref{eqn:er}. To address this, We resort to the upper and lower bounds of EC, as given in the following theorem.


\begin{theorem}\label{th_cap}
The upper and lower bounds of the asymptotic EC are given by
\begin{align}\label{bound_ER}
    \bar{C}_{\rm asy}^{\rm low}<\bar{C}_{\rm asy}\le \bar{C}_{\rm asy}^{\rm upp}.
\end{align}
where
\begin{equation}\label{eqn:upp_ec}
    \bar{C}_{\rm asy}^{\rm upp} = \int_0^\infty \log_2\left(1+\bar{\gamma}\left( l_{N,t}+\Upsilon  a_{N,t} \right)\right)e^{-t} {\rm{d}} t,
\end{equation}
\begin{align}\label{lowerbound_ER}
      \bar{C}_{\rm asy}^{\rm low}=&\int\limits_0^{\infty}\log_2(1+\bar{\gamma}l_{N,t})e^{-t} {\rm{d}} t+\frac{\Upsilon \bar{\gamma} }{\ln 2}\int\limits_0^{\infty}  \frac{a_{N,t}}{1+\bar{\gamma}l_{N,t}} e^{-t}{\rm{d}} t \notag\\
     &- \frac{(6\Upsilon ^2+\pi^2)\bar{\gamma}^2}{12\ln 2}\int_0^{\infty}  \frac{a_{N,t}^2}{(1+\bar{\gamma}l_{N,t})^2}e^{-t} {\rm{d}} t.
\end{align}    
\end{theorem}
\begin{proof}
    Please see Appendix \ref{sec_th_2}.
\end{proof}

By combining the asymptotic results of $l_{N,t}$ and $a_{N,t}$ in Theorem \ref{th_h_dis} with Theorem \ref{th_cap}, the dominant term approximation shows that both the upper and the lower bounds of the asymptotic EC take the same form, i.e., $\bar{C}_{\rm asy}^{\rm upp} ( \bar{C}_{\rm asy}^{\rm low}) \propto \log_2(1+\bar{\gamma}\sigma^2(1-\mu^2)\ln N)$. The squeeze theorem thus establishes that the EC increases according to a double-logarithmic scaling law with respect to the number of antenna ports in FAS. In addition, the asymptotic results also disclose that the correlation parameter $\mu$ negatively impacts the EC. 


In Theorem \ref{th_h_dis}, the mode of the equivalent channel gain $g$ increases logarithmically with the number of ports. This observation prompts a natural question: does the mode of the mutual information (i.e., \({\rm Mo}(\mathcal I|g_0)\)) exhibit a similar logarithmic scaling?
To investigate this, we derive the conditional PDF of the mutual information \(\mathcal{I}\) given \(g_0\). Using the transformation \(g = (2^{\mathcal{I}} - 1)/\bar{\gamma}\), we arrive at
\begin{align}
    f_{\mathcal{I}|g_0}(z) = \frac{2^z \ln 2}{\bar{\gamma}} f_{g|g_0}\left( \frac{2^z - 1}{\bar{\gamma}} \right).
\end{align}

The mode of \(f_{\mathcal{I}|g_0}(z)\) is found by setting its first derivative to zero, which leads to the condition
\begin{align}\label{z_pdf_der}
    f^\prime_{g|g_0}\left( \frac{2^{{\rm Mo}(\mathcal I|g_0)} - 1}{\bar{\gamma}} \right) = -\frac{\bar{\gamma}}{2^{{\rm Mo}(\mathcal I|g_0)}} f_{g|g_0}\left( \frac{2^{{\rm Mo}(\mathcal I|g_0)} - 1}{\bar{\gamma}} \right).
\end{align}

Although a closed-form solution for \({\rm Mo}(\mathcal I|g_0)\) is intractable, \eqref{z_pdf_der} reveals that \(f^\prime_{g|g_0}(\cdot) < 0\). This implies the inequality \({(2^{{\rm Mo}(\mathcal I|g_0)} - 1)}/{\bar{\gamma}} > l_{N,t}\), and consequently, the mode satisfies
\begin{align}\label{eqn:mo_I}
    {\rm Mo}(\mathcal I|g_0) &> \log_2(1 + \bar{\gamma} l_{N,t}) \notag\\
    & \simeq  \log_2 \left( 1 + \bar{\gamma} \sigma^2 (1 - \mu^2) \ln N \right)\propto \ln\ln N,
\end{align}
where the second step holds by using the asymptotic result \(l_{N,t} \simeq \sigma^2 (1 - \mu^2) \ln N\). This result demonstrates that the mode of the mutual information indeed follows a double-logarithmic scaling law with the number of fluid antenna ports. 

\section{Numerical Results}\label{sec_num}
In this section, numerical results are presented to demonstrate our analytical results. 
For illustration, the FAS parameters are configured as follows  unless otherwise specified: $\bar{\gamma}=10 $ dB, $W=0.3$, $\sigma^2=1$, $N=15$ and $R=4$ bit/s/Hz. The labels ``Asy.'' and ``Sim.'' refer to the asymptotic and the simulated results, respectively.

\begin{figure*}[ht]
    \centering
    \subfigure[]{
        \includegraphics[width=0.23\linewidth]{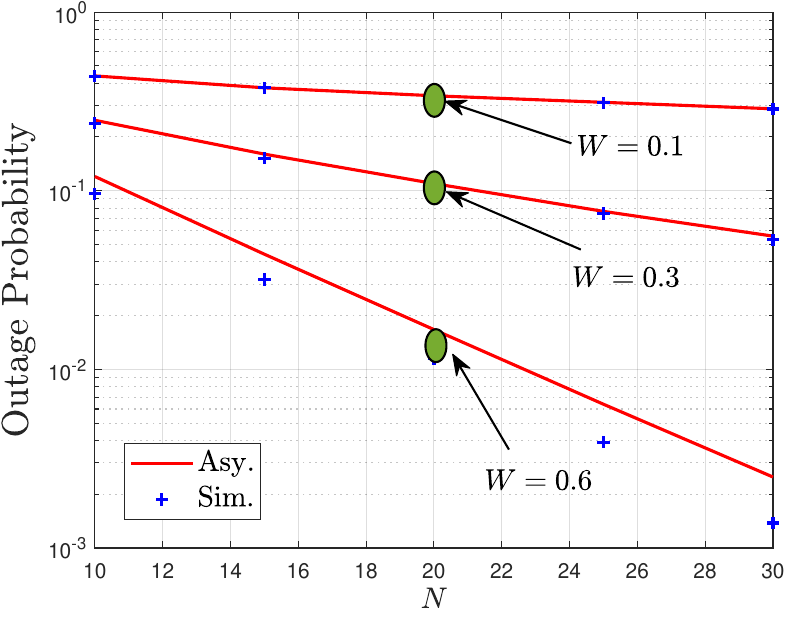}
        \label{op_N}
    }
    \subfigure[]{
        \includegraphics[width=0.23\linewidth]{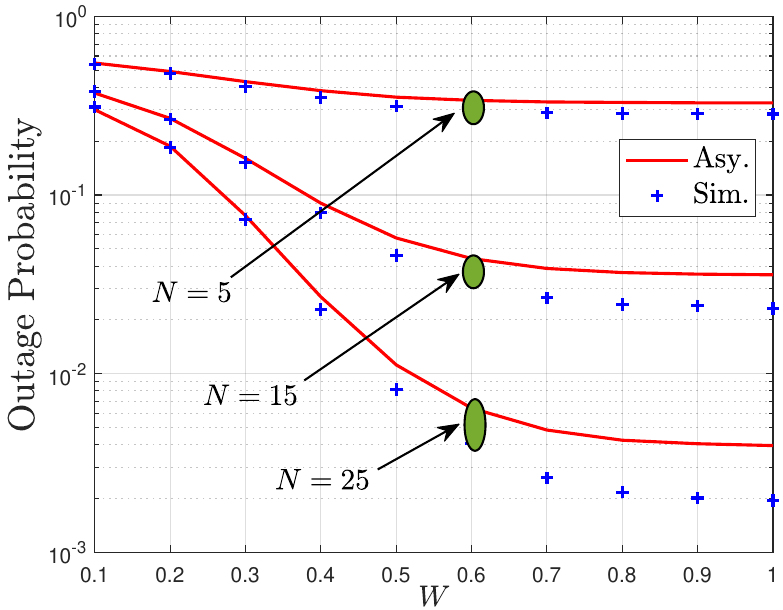}
        \label{op_cor}
    }
    \subfigure[]{
        \includegraphics[width=0.23\linewidth]{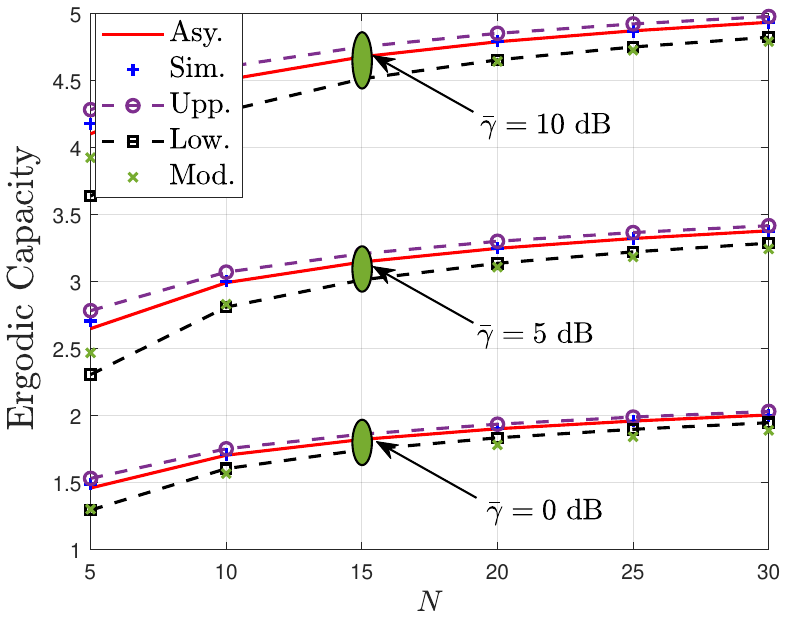}
        \label{ec_N}
    }
    \subfigure[]{
        \includegraphics[width=0.23\linewidth]{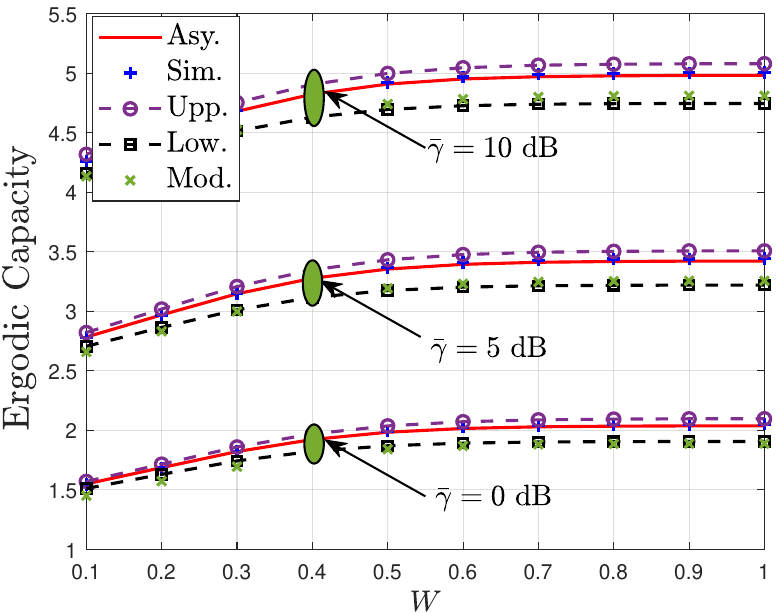}
        \label{ec_mu}
    }
   \caption{(a) The OP versus $N$; (b) The OP versus $W$; (c) The EC versus $N$; (d) The EC versus $W$.}
  \label{fig:all_results}
  \vspace{-0.5cm}
\end{figure*}


As shown in Fig.~\ref{op_N}, the OP is plotted against the number of antenna ports $N$. Clearly, the asymptotic results align closely with the simulated ones, validating our asymptotic analysis. Moreover, the OP exhibits an approximate exponential decay with $N$ owing to the extra spatial degree of freedom offered the FAS. In Fig.~\ref{op_cor}, we investigate the impact of the normalized length $W$ on the OP. The results depict that the outage performance improves with increasing $W$, particularly in the small $W$ regime. This notable gain is attributed to the low spatial correlation between antenna ports, as quantified by \eqref{mu_W}. Moreover, there exists an outage performance floor if $W \ge 0.8$, which coincides with \cite[Eq. 10]{wong2022closed}. 



The effects of $N$ and $W$ upon the EC are illustrated in Figs.~\ref{ec_N} and \ref{ec_mu}, respectively, where ``Upp.'' and ``Low.'' denote the upper and the lower bounds of the asymptotic EC, respectively, the expectation of \({\rm Mo}(\mathcal I|g_0)\) (i.e., ${\mathbb E_{{g_0}}}\left({\rm Mo}(\mathcal I|g_0)\right)$) over $g_0$ is labeled as ``Mod.''. From both figures, the asymptotic results agree well with the simulated ones, and accurately lie between its derived upper and lower bounds. Moreover, the results of ``Mod'' match well with the lower bound. This is because ${\mathbb E_{{g_0}}}\left({\rm Mo}(\mathcal I|g_0)\right)$ is factually the first term of the lower bound in \eqref{lowerbound_ER} and the other two terms become negligible for a large $N$. Besides, Fig. \ref{ec_N} successfully illustrates the predicted double-logarithmic growth trend of the EC with respect to $N$, thereby validating the scaling law of EC. 

\section{Conclusion}\label{sec_con}
This paper has applied EVT to perform an asymptotic analysis of the FAS, deriving the asymptotic OP and EC under a large number of antenna ports. The key findings conclusively show that the OP decays approximately exponentially with the number of antenna ports, while the EC follows a double-logarithmic scaling law. These scaling behaviors are further corroborated by the logarithmic-like growth of the modes of both the channel gain and the mutual information. Furthermore, the asymptotic analysis confirms the negative influence of spatial correlation on the OP and the EC. All theoretical findings are supported by our numerical observations.

\appendices
\section{Proof of Theorem \ref{th_h_dis}}\label{sec_th_1}
To prove that the distribution of $g$ converges to the Gumbel distribution as $N\to \infty$, it suffices to show that \eqref{le2_lim} holds according to Lemma \ref{le2_ind}. To proceed, the asymptotic expressions of $f_{|h_n|^2|g_0}(x_n|t)$ and $1-F_{|h_n|^2|g_0}(x_n|t)$ as $x_n\to \infty$ should be derived. Putting the asymptotic expansion of the modified Bessel function $I_0(x)\simeq \frac{e^x}{\sqrt{2\pi x}}$ as $x\to \infty$\cite[Eq.9.7.1]{abramowitz1968handbook} into \eqref{h_pdf} results in the asymptotic expression of $f_{|h_n|^2|g_0}(x_n|t)$ as
\begin{align}\label{pdf_asy}
    &f_{|h_n|^2|g_0}\left(\left.x_n\right|t\right)\notag\\
    &\simeq\frac{(\sigma^2{\mu}^2x_nt)^{-\frac1{4}}}{{\sqrt{4\pi \sigma^2(1-\mu^2)}}}\exp\left(-\frac{(\sqrt{x_n}-\sigma{\mu}\sqrt{t})^2}{\sigma^2(1-\mu^2)}\right),
\end{align}
where $\simeq$ means ``is proportional to''. In addition, the asymptotic result of $1-F_{|h_n|^2|T}$ can be obtained as \eqref{cdf_asy_2}, as shown at the top of the next page, where step (a) holds by employing the asymptotic approximation $Q_1(\alpha,\beta)\simeq \sqrt{\frac{\beta}{\alpha}}Q(\beta-\alpha)$ as $\beta\to\infty$ \cite{temme1993asymptotic}, step (b) holds by using $Q(x)=\frac1{2}{\rm{erfc}}(\frac{x}{\sqrt{2}})$, step (c) holds by using ${\rm{erfc}}(x)\simeq \frac1{x\sqrt{\pi}}e^{-x^2}$, $Q(\cdot)$ denotes the Guassian Q-function, and ${\rm{erfc}}(\cdot)$ represents complementary error function.
\begin{figure*}[!t]
\begin{multline}
    1-F_{|h_n|^2|g}\left(\left.x_n\right|t\right)=Q_1\left(\sqrt{\frac{2\mu^2t}{1-\mu^2}},\sqrt{\frac{2x_n}{\sigma^2(1-\mu^2)}}\right)\stackrel{(a)}{\simeq} \left(\frac{x_n}{\sigma^2\mu^2t}\right)^{\frac{1}{4}}Q\left(\sqrt{\frac{2x_n}{\sigma^2(1-\mu^2)}}-\sqrt{\frac{2\mu^2t}{1-\mu^2}}\right)\\
\stackrel{(b)}{=} \left(\frac{x_n}{\sigma^2\mu^2t}\right)^{\frac1{4}}\frac1{2}{\rm{erfc}}\left({\frac{\sqrt{x_n}-\sqrt{\sigma^2\mu^2t}}{\sqrt{\sigma^2(1-\mu^2)}}}\right)
\stackrel{(c)}{\simeq}\frac{\sqrt{\sigma^2(1-\mu^2)}}{2\sqrt{\pi}}({\sigma^2\mu^2tx_n})^{-\frac1{4}}\exp\left(-{\frac{(\sqrt{x_n}-\sqrt{\sigma^2\mu^2t})^2}{{\sigma^2(1-\mu^2)}}}\right).\label{cdf_asy_2}
\end{multline}
\hrulefill
  \vspace{-0.5cm}
\end{figure*}
Therefore, by putting \eqref{pdf_asy} and \eqref{cdf_asy_2} into \eqref{le2_lim}, we can get 
\begin{align}
    \lim\limits_{x_n\to {\infty}} \frac{\rm d}{{\rm d} x_n}\frac{1-F_{|h_n|^2|g_0}\left(\left.x_n\right|t\right)}{f_{|h_n|^2|g_0}\left(\left.x_n\right|t\right)}= \frac{\rm d \left(\sigma^2(1-\mu^2)\right) }{{\rm d} x_n}=0.
\end{align}
According to Lemma \ref{le2_ind}, \eqref{eqn:f_cdf} and \eqref{eqn:f_pdf} can be directly obtained.

With the location and the scale parameters given by Lemma \ref{le2_ind}, the mean, the variance, second moment, and third moment of Gumbel distribution can be obtained by using \cite{chattamvelli2022gumbel}. Furthermore, by plugging \eqref{pdf_asy} and \eqref{cdf_asy_2} into the expressions of the location and the scale parameters, the asymptotics of $l_{N,t}$ and $a_{N,t}$ are derived as $l_{N,t}=\left(\sqrt{\sigma^2(1-\mu^2)\ln N}+\sqrt{\sigma^2\mu^2t}\right)^2+o( \ln N)\simeq \sigma^2(1-\mu^2)\ln N$ and $a_{N,t}\simeq \sqrt{4\pi\mu\sigma^4(1-\mu^2)^{3/2}}(t\ln N)^{\frac14}$ as $N$ approaches to $\infty$, where $Q_1^{-1}$ denotes the inverse first-order Marcum Q-function and $o(\cdot)$ represents the little-O notation. 

Moreover, to get the mode of the distribution of $g$ conditioned on $g_0$, it suffices to find the maximum value of the PDF $f_{g|g_0}(x|t)$. Towards this end, taking the first and second derivatives of $f_{g|g_0}(x|t)$ yields 
\begin{align}
f_{g|g_0}^\prime(x|t)
=&\frac{1}{a_{N,t}^2}e^{-e^{-\frac{x-l_{N,t}}{a_{N,t}}}-\frac{x-l_{N,t}}{a_{N,t}}}\bigl(e^{-\frac{x-l_{N,t}}{a_{N,t}}}-1\bigr)\label{pdf_t_der},\\
f_{g|g_0}^{\prime\prime}(x|t)=&\frac{1}{a_{N,t}^3}e^{-e^{-\frac{x-l_{N,t}}{a_{N,t}}}-\frac{x-l_{N,t}}{a_{N,t}}}\notag\\
&\times\left(e^{-\frac{2(x-l_{N,t})}{a_{N,t}}} -3e^{-\frac{x-l_{N,t}}{a_{N,t}}}+1\right).\label{pdf_t_der2}
\end{align}
Clearly, we have $f_{g|g_0}^\prime(l_{N,t}|t)=0$ and $f_{g|g_0}^{\prime\prime}(l_{N,t}|t)=-{1}/{a_{N,t}^3}<0$. Accordingly, the mode ${\rm Mo}(g|g_0) =  \arg\max_{x} f_{g|g_0}(x|t) = l_{N,t}$ is substantiated. Furthermore, Van Zwet condition indicates that the mean of the distribution $f_{g|g_0}(x|t)$ is larger than its mode, i.e., $\mathbb E(g|g_0) \ge {\rm Mo}(g|g_0)$. The proof is thus completed.

\section{Proof of Theorem \ref{th_cap}}\label{sec_th_2}



To prove Theorem \ref{th_cap}, it is imperative to obtain the upper and the lower bounds of $\Bbb{E}_{g|g_0}(\log_2(1+\bar{\gamma}x)|t)$ according to \eqref{eqn:er_der}.

\subsection{Upper Bound}
By noticing the concavity of logarithmic functions, leveraging Jensen's inequality produces the upper bound of the conditional EC as
\begin{align}\label{eqn:er_upp_g}
 \Bbb{E}_{g|g_0}\left( \log_2\left(1+\bar{\gamma}x\right)|t\right) \le \log_2\left(1+\bar{\gamma}\Bbb{E}\left(g|g_0=t\right) \right).
\end{align}
By substituting $\Bbb{E}_{g|{g_0}}(g|t)=l_{N,t}+\Upsilon  a_{N,t}$ in Theorem \ref{th_h_dis} into \eqref{eqn:er_der}, the upper bound of the asymptotic EC can be obtained as \eqref{eqn:upp_ec}.

\subsection{Lower Bound}
The lower bound \eqref{lowerbound_ER} is proved by using Taylor expansion of $\phi(g) \triangleq \log_2(1+\bar{\gamma}g)$. In particular, we expand the Taylor series of $\phi(g)$ around the mode ${\rm Mo}(g|g_0) = l_{N,t}$ as
\begin{multline}\label{low_bound}
    \phi(g)=\phi(l_{N,t})+\phi^\prime(l_{N,t})(g-l_{N,t})\\
    +\phi^{\prime\prime}(l_{N,t})(g-l_{N,t})^2/2!+\phi^{\prime\prime\prime}(\xi)(g-l_{N,t})^3/3!.
\end{multline}
where $\xi \in [l_{N,t},g]$. 
Taking the conditional expectation with respect to $g$ given $g_0$ on both sides of \eqref{low_bound} gives
\begin{align}\label{low_bound_2}
    \Bbb{E}_{g|g_0}(\phi(g)|t)=& \phi(l_{N,t})+\phi^\prime(l_{N,t})(\Bbb{E}_{g|g_0}(g|t)-l_{N,t})\notag\\
    &+\phi^{\prime\prime}(l_{N,t})\Bbb{E}_{g|g_0}((g-l_{N,t})^2|t)/2!\notag\\
    &+\phi^{\prime\prime\prime}(\tilde \xi)\Bbb{E}_{g|g_0}((g-l_{N,t})^3|t)/3!,
\end{align}
where $\tilde \xi \in [l_{N,t},g]$. By using the result of the first three moments in Theorem \ref{th_h_dis}, $\Bbb{E}_{g|g_0}((g-l_{N,t})^2|t)$ and $\Bbb{E}_{g|g_0}((g-l_{N,t})^3|t)$ are obtained as
\begin{align}\label{eqn:2or_g0}
    \Bbb{E}_{g|g_0}((g-l_{N,t})^2|t)
&=a_{N,t}^2\left(\Upsilon ^2+\frac{\pi^2}{6}\right),
\end{align}
\begin{align}\label{eqn:3or_g0}
    &\Bbb{E}_{g|g_0}((g-l_{N,t})^3|t)\notag\\
    &=\Bbb{E}_{g|g_0}(g^3|t)-3l_{N,t}\Bbb{E}_{g|g_0}(g^2|t)+3l_{N,t}^2\Bbb{E}_{g|g_0}(g|t)-l_{N,t}^3\notag\\
&=a_{N,t}^3\left(\Upsilon ^3+\frac{\Upsilon \pi^2}{2}+2\zeta(3)\right)>0.
\end{align}
On the basis of $\phi^{\prime\prime\prime}(l_{N,t})=2\gamma^3/((1+\gamma l_{N,t})^3\ln2)>0$ and \eqref{eqn:3or_g0}, substituting $\Bbb{E}_{g|{g_0}}(g|t)=l_{N,t}+\Upsilon  a_{N,t}$ in Theorem \ref{th_h_dis} and \eqref{eqn:2or_g0} leads to the lower bound of $\Bbb{E}_{g|g_0}(\phi(g)|t)$ as
\begin{align}\label{low_bound_3}
    \Bbb{E}_{g|g_0}(\phi(x)|t)>&\phi(l_{N,t})+\Upsilon  a_{N,t} \phi^\prime(l_{N,t})\notag\\
&+\frac{a_{N,t}^2}{2}\left(\Upsilon ^2+\frac{\pi^2}{6}\right)\phi^{\prime\prime}(l_{N,t}).
\end{align}
As a consequence, putting \eqref{low_bound_3} into \eqref{eqn:er_der} leads to the lower bound of the asymptotic EC as \eqref{lowerbound_ER}.

\bibliographystyle{IEEEtran}
\bibliography{IEEEabrv,references}

\end{document}